\PassOptionsToPackage{unicode}{hyperref}
\PassOptionsToPackage{hyphens}{url}
\PassOptionsToPackage{dvipsnames,svgnames,x11names}{xcolor}
\documentclass[12pt]{article}
\usepackage{textcomp}
 \usepackage{amssymb}
\usepackage{amsmath}
\usepackage{amsthm}
\newtheorem{theorem}{Theorem}
\newtheorem{assumption}{Assumption}
\newtheorem{proposition}{Proposition}
\newtheorem{lemma}{Lemma}

\newtheorem{definition}{Definition}

\newtheorem*{remark}{Remark}
\usepackage{xr}
\usepackage{enumerate}
\usepackage{setspace}
\usepackage{natbib}
\usepackage{iftex}
\ifPDFTeX
  \usepackage[T1]{fontenc}
  \usepackage[utf8]{inputenc}
  \usepackage{textcomp} 
\else 
  \usepackage{unicode-math}
  \defaultfontfeatures{Scale=MatchLowercase}
  \defaultfontfeatures[\rmfamily]{Ligatures=TeX,Scale=1}
\fi
\usepackage{lmodern}
\ifPDFTeX\else  
\fi
\IfFileExists{upquote.sty}{\usepackage{upquote}}{}
\IfFileExists{microtype.sty}{
  \usepackage[]{microtype}
  \UseMicrotypeSet[protrusion]{basicmath} 
}{}
\makeatletter
\@ifundefined{KOMAClassName}{
  \IfFileExists{parskip.sty}{%
    \usepackage{parskip}
  }{
    \setlength{\parindent}{0pt}
    \setlength{\parskip}{6pt plus 2pt minus 1pt}}
}{
  \KOMAoptions{parskip=half}}
\makeatother
\usepackage{xcolor}
\makeatletter
\ifx\paragraph\undefined\else
  \let\oldparagraph\paragraph
  \renewcommand{\paragraph}{
    \@ifstar
      \xxxParagraphStar
      \xxxParagraphNoStar
  }
  \newcommand{\xxxParagraphStar}[1]{\oldparagraph*{#1}\mbox{}}
  \newcommand{\xxxParagraphNoStar}[1]{\oldparagraph{#1}\mbox{}}
\fi
\ifx\subparagraph\undefined\else
  \let\oldsubparagraph\subparagraph
  \renewcommand{\subparagraph}{
    \@ifstar
      \xxxSubParagraphStar
      \xxxSubParagraphNoStar
  }
  \newcommand{\xxxSubParagraphStar}[1]{\oldsubparagraph*{#1}\mbox{}}
  \newcommand{\xxxSubParagraphNoStar}[1]{\oldsubparagraph{#1}\mbox{}}
\fi
\makeatother

\usepackage{longtable,booktabs,array}
\usepackage{calc} 
\usepackage{etoolbox}
\makeatletter
\patchcmd\longtable{\par}{\if@noskipsec\mbox{}\fi\par}{}{}
\makeatother
\IfFileExists{footnotehyper.sty}{\usepackage{footnotehyper}}{\usepackage{footnote}}
\makesavenoteenv{longtable}
\usepackage{graphicx}
\makeatletter
\def\maxwidth{\ifdim\Gin@nat@width>\linewidth\linewidth\else\Gin@nat@width\fi}
\def\maxheight{\ifdim\Gin@nat@height>\textheight\textheight\else\Gin@nat@height\fi}
\makeatother
\setkeys{Gin}{width=\maxwidth,height=\maxheight,keepaspectratio}
\makeatletter
\def\fps@figure{htbp}
\makeatother

\makeatletter
\@ifpackageloaded{caption}{}{\usepackage{caption}}
\AtBeginDocument{%
\ifdefined\contentsname
  \renewcommand*\contentsname{Table of contents}
\else
  \newcommand\contentsname{Table of contents}
\fi
\ifdefined\listfigurename
  \renewcommand*\listfigurename{List of Figures}
\else
  \newcommand\listfigurename{List of Figures}
\fi
\ifdefined\listtablename
  \renewcommand*\listtablename{List of Tables}
\else
  \newcommand\listtablename{List of Tables}
\fi
\ifdefined\figurename
  \renewcommand*\figurename{Figure}
\else
  \newcommand\figurename{Figure}
\fi
\ifdefined\tablename
  \renewcommand*\tablename{Table}
\else
  \newcommand\tablename{Table}
\fi
}
\@ifpackageloaded{float}{}{\usepackage{float}}
\floatstyle{ruled}
\@ifundefined{c@chapter}{\newfloat{codelisting}{h}{lop}}{\newfloat{codelisting}{h}{lop}[chapter]}
\floatname{codelisting}{Listing}

\makeatother
\makeatletter
\@ifpackageloaded{caption}{}{\usepackage{caption}}
\@ifpackageloaded{subcaption}{}{\usepackage{subcaption}}
\makeatother

\ifLuaTeX
  \usepackage{selnolig}  
\fi
\usepackage[]{natbib}
\usepackage{bookmark}

\IfFileExists{xurl.sty}{\usepackage{xurl}}{} 
\hypersetup{
  pdftitle={Robust Estimation under Spurious Correlation via Counterfactual Augmentation},
  pdfauthor={Author 1; Author 2},
  pdfkeywords={3 to 6 keywords, that do not appear in the title},
  colorlinks=true,
  linkcolor={blue},
  filecolor={Maroon},
  citecolor={Blue},
  urlcolor={Blue},
  pdfcreator={LaTeX via pandoc}}

\usepackage{amsmath} 

\newcommand{\anon}{1}

\begin{document}
\setlength{\parskip}{3pt plus 1pt minus 1pt}

\def\spacingset#1{\renewcommand{\baselinestretch}%
{#1}\small\normalsize} \spacingset{1}


\if1\anon
{
  \title{\bf Robust Estimation under Spurious Correlation via Statistical Counterfactual Augmentation}
  \author{Yi-Hui Zhou\textsuperscript{*} \hspace{.2cm}\\
    Department of Statistics, North Carolina State University\\George Sun \\High Dimensional Predictive Biology Lab, North Carolina State University}
  \maketitle
} \fi

\if0\anon
{
  \bigskip
  \bigskip
  \bigskip
  \begin{center}
    {\LARGE\bf  Robust Estimation under Spurious Correlation via Statistical Counterfactual Augmentation}\hspace{.2cm}\\
\end{center}
  \medskip
} \fi

\bigskip
\bigskip
\begin{abstract}
Spurious correlations between nuisance features and outcomes present a fundamental challenge for statistical learning and inference. Counterfactual Data Augmentation (CDA), the practice of resampling nuisance features while preserving features of interest, has emerged as a popular strategy to mitigate such correlations, yet its statistical properties, and its reliance on an untestable validity assumption, remain largely unexplored. We provide a statistical treatment of CDA, termed Statistical CDA (SCDA), with two components. First, casting SCDA as a simulated-moment M-estimator, we establish consistency, asymptotic normality with a sandwich covariance and a finite-augmentation-draw variance correction, oracle efficiency for the target parameter, and the equivalence of iterative and one-shot procedures via stochastic approximation, yielding valid large-sample inference where previously none was available. Second, because the augmentation's validity cannot be tested from the observed data when features are entangled (a fact we prove), we develop a sensitivity analysis that bounds the residual nuisance--outcome dependence by a parameter $\Gamma$, giving sharp identified intervals for the target with valid confidence intervals and a breakdown frontier; the point-identified theory is the $\Gamma=0$ special case. Simulations confirm the efficiency gains under validity and the robustness bounds under possible invalidity, placing a widely used augmentation method on a transparent inferential footing.
\end{abstract}

\noindent%
\vfill

\newpage
\spacingset{1.75} 

\section{Introduction}\label{sec-intro}

Counterfactual data augmentation (CDA) has become a widely used strategy in machine learning for mitigating spurious correlations and improving the robustness of predictive models. The basic idea is to generate synthetic examples in which nuisance features are perturbed while causal features are preserved, and then to retrain models on the combined original and augmented data. Empirical studies have demonstrated that CDA can improve generalization and interpretability in settings such as natural language processing and computer vision, where strong correlations between labels and nuisance variables are common \citep{lu2020gender,kaushik2020learning,plyler2021cda}.  

Despite its empirical success, CDA currently lacks a rigorous statistical foundation. Existing work has primarily been heuristic, focusing on implementation and empirical performance, with little investigation of CDA as an estimation procedure. In particular, there has been no systematic study of its large-sample behavior, efficiency, or robustness. This gap motivates the present article, in which we develop a statistical theory of CDA. To distinguish our contribution from the purely algorithmic literature, we refer to our framework as \emph{Statistical CDA}.  

The purpose of this article is to provide a formal foundation for Statistical CDA. We show that one-shot CDA estimators can be cast as M-estimators with a nested expectation structure induced by the augmentation mechanism. This formulation allows us to extend classical asymptotic theory \citep{van2000asymptotic,vaart1998asymptotic} to the CDA setting. In particular, we establish conditions under which Statistical CDA estimators are consistent and asymptotically normal, and we derive their asymptotic covariance structure. These results yield valid large-sample inference for parameters estimated under CDA, a topic that, to the best of our knowledge, has not previously been rigorously addressed.  

Our analysis also clarifies efficiency and robustness: Statistical CDA attains vanishing excess risk
under the augmentation-induced target at the parametric rate, whereas empirical risk minimization
incurs a fixed positive excess risk under strong spurious correlation, and its influence function
remains bounded under mild augmentation misspecification. Iterative and adaptive procedures, though not
strictly M-estimators, are analyzed as stochastic approximation schemes and shown to share the
one-shot limiting distribution, formalizing heuristic arguments from the machine learning literature
\citep{plyler2025icda}.

The results above all presume that the augmentation is \emph{valid}: resampling the nuisance renders it independent of the outcome given the features of interest. This assumption is not testable when the two are entangled in the observed representation, which is exactly the regime, such as image data, in which CDA is most used. Rather than treat validity as given, we make it the object of analysis: we introduce a sensitivity model that bounds the residual nuisance--outcome dependence by a parameter $\Gamma$, derive sharp identified intervals for the target with valid confidence intervals, and characterize a breakdown frontier quantifying how much invalidity a substantive conclusion can withstand (Section~\ref{sec:sensitivity}). The point-identified theory is recovered at $\Gamma=0$. This turns an untestable premise into a graded, reportable robustness statement that complements the estimation theory above.

In summary, the article makes four contributions. (i) An estimand identity: the augmented population risk equals the risk under the nuisance-marginalized target $P_{\mathcal{A}}$ (Lemma~\ref{lem:identity}), so the augmentation \emph{defines} the estimand; this resolves what augmentation-based training estimates. (ii) Large-sample inference for that estimand: consistency, asymptotic normality with a sandwich covariance, and a finite-augmentation-draw correction without which nominal coverage fails (Proposition~\ref{prop:mcvar}). (iii) An augmentation-invariance and oracle-efficiency theorem: every valid augmentation yields the same asymptotic variance for the target parameter, equal to that of the oracle with direct access to $X_1$ (Theorem~\ref{thm:oracle}). (iv) A sensitivity analysis for augmentation validity, whose untestability from the observed data we prove (Proposition~\ref{prop:untestable}), with sharp identified intervals and a breakdown frontier (Section~\ref{sec:sensitivity}).

In addition to the theoretical results, we investigate the finite-sample behavior of Statistical CDA through simulation studies. These confirm the variance reduction predicted by our theory and demonstrate that inference based on the derived asymptotic covariance matrix achieves nominal coverage, and a controlled entangled-representation design isolates the regime in which augmentation is necessary. Taken together, our results connect a widely used augmentation method with the established theory of M-estimation and stochastic approximation, thereby placing CDA on rigorous statistical footing.

\section{Methodology}\label{sec:method}

We now formalize the framework of \emph{Statistical Counterfactual Data Augmentation} (SCDA). 
This section introduces the notation, defines the SCDA estimator, and compares it with related approaches. 
Our formulation makes clear the departures from standard empirical risk minimization and motivates the asymptotic analysis in Section~\ref{sec:theory}.  

\subsection{Framework and Notation}

Let $\{(X_i,Y_i): i=1,\ldots,n\}$ denote an independent and identically distributed (i.i.d.) sample from an unknown distribution $P_0$ on $\mathcal{X}\times\mathcal{Y}$, where $X_i \in \mathcal{X}$ are covariates and $Y_i \in \mathcal{Y}$ are outcomes. We suppose that each $X_i$ can be decomposed as $X_i = (X_{i1},X_{i2})$, where $X_{i1}$ represents features of primary or causal interest, and $X_{i2}$ represents nuisance or spurious features that may be correlated with $Y_i$.  

Let $f_\theta: \mathcal{X} \to \mathcal{Y}$ be a parametric prediction model indexed by $\theta \in \Theta \subset \mathbb{R}^p$, and let $\ell(y,f_\theta(x))$ denote a loss function (e.g., logistic loss for classification or squared loss for regression). The conventional empirical risk minimizer (ERM) is defined as
\[
\hat\theta_{\mathrm{ERM}} = \arg\min_{\theta \in \Theta} \; \frac{1}{n}\sum_{i=1}^n \ell(Y_i, f_\theta(X_i)).
\]

\subsection{The SCDA Estimator}

The central idea of SCDA is to generate counterfactual covariates by perturbing nuisance components while preserving causal ones, thereby attenuating the influence of spurious associations. Our use of ``counterfactual'' follows the machine learning literature \citep{kaushik2020learning} and refers to data generated by perturbing nuisance features. Under the assumption that $X_2$ has no causal effect on $Y$ given $X_1$, our augmented samples correspond to counterfactuals in the structural causal model (SCM) sense \citep{pearl2009causality}. Specifically, if $Y = f(X_1, X_2, U)$ is a structural equation in which $X_2$ does not appear in $f$ (i.e., $X_2$ has no causal effect on $Y$ given $X_1$), then augmented samples $\tilde{X} = (X_1, \tilde{X}_2)$ with $\tilde{X}_2 \perp Y \mid X_1$ correspond to the counterfactual ``what would $Y$ be if $X_2$ were $\tilde{X}_2$ instead of $X_2$?'' The three construction strategies in Section~\ref{sec:practical_guidance} (permutation-based, generative, deletion or masking) produce such counterfactuals when they preserve $X_1$ and break the association between $X_2$ and $Y$ given $X_1$, under this same causal assumption. This causal assumption is for \emph{interpretation} (SCM counterfactuals) and \emph{motivation} (why attenuate $X_2$). The \emph{statistical} validity of the SCDA estimator (consistency, asymptotic normality in Section~\ref{sec:theory}) relies on the formal conditions stated there (Assumptions~\ref{ass:loss}--\ref{ass:nonsing}: loss regularity, augmentation, compactness, identifiability, nonsingularity); the causal assumption is \emph{not} among them. Absent this causal assumption, the augmented samples remain heuristic perturbations in the sense of the machine learning (e.g., computer vision and NLP) literature. Formally, let $\mathcal{A}(x)$ denote a conditional distribution over augmented covariates $\tilde X$ given $X=x$. The distribution $\mathcal{A}(x)$ preserves $X_1$ but perturbs $X_2$ so as to weaken or eliminate its association with $Y$. Examples include swapping nuisance features across samples \citep{zeng2020counterfactual} or generating synthetic nuisance features conditioned only on $X_1$.

Given an augmentation mechanism $\mathcal{A}$, the augmented population risk is
\[
R(\theta) = \mathbb{E}_{(X,Y)\sim P_0} \; \mathbb{E}_{\tilde X \sim \mathcal{A}(X)} \big[ \ell(Y, f_\theta(\tilde X)) \big],
\]
with empirical analogue
\[
R_n(\theta) = \frac{1}{n}\sum_{i=1}^n \; \mathbb{E}_{\tilde X_i \sim \mathcal{A}(X_i)} \big[ \ell(Y_i, f_\theta(\tilde X_i)) \big].
\]
We define the population minimizer as $\theta_0 = \arg\min_{\theta \in \Theta} R(\theta)$, which is the unique solution to $\Psi(\theta_0) = 0$ under Assumption~\ref{ass:ident} (see Section~\ref{sec:theory}).

We define the SCDA estimator as
\[
\hat\theta_{\mathrm{SCDA}} = \arg\min_{\theta \in \Theta} R_n(\theta).
\]
Under differentiability of $R_n(\theta)$ and when $\hat\theta_{\mathrm{SCDA}}$ is a stationary point, this is equivalent to solving the estimating equation
\[
\Psi_n(\theta) = \frac{1}{n}\sum_{i=1}^n \; \mathbb{E}_{\tilde X_i \sim \mathcal{A}(X_i)} \nabla_\theta \ell(Y_i, f_\theta(\tilde X_i)) = 0.
\]
Assumption~\ref{ass:ident} (introduced in Section~\ref{sec:theory}) ensures that the stationary point is the unique global minimizer.

Thus $\hat\theta_{\mathrm{SCDA}}$ is an M-estimator with a \emph{nested expectation structure}, reflecting both sampling variability from $(X_i,Y_i)$ and augmentation variability from $\tilde X_i \sim \mathcal{A}(X_i)$.  
The feature distinguishing SCDA from classical M-estimation is the \emph{nested expectation}: the estimating function is an inner expectation over augmented covariates $\tilde X \sim \mathcal{A}(X)$, which introduces an additional variance component (Section~\ref{sec:montecarlo}) and requires analysis beyond the standard theory. In the one-shot estimator the augmentation $\mathcal{A}$ is fixed and does not depend on $\theta$; the adaptive case, in which the augmentation evolves with the current estimate, is treated in Section~\ref{sec:iterative}. A score identity in the Supplement makes the role of $\theta$ precise: when $\mathcal{A}$ is $\theta$-free, the estimating function $\Psi_n$ is exactly the gradient of $R_n$, whereas a $\theta$-dependent augmentation contributes an additional covariance term between the loss and the augmentation score. The analogous quantity at the level of the estimating equation, the sensitivity $G_0$ of the population score to the augmentation index, must vanish at $\theta_0$ in the adaptive case; this is the orthogonality condition (Assumption~\ref{ass:orth}).

\subsection{Comparison with Related Estimators}

It is instructive to contrast SCDA with familiar alternatives.  
The ERM estimator $\hat\theta_{\mathrm{ERM}}$ minimizes risk based solely on observed data, without augmentation.  
Inverse probability weighting (IPW) estimators \citep{horvitz1952generalization,rosenbaum1983central} address confounding by reweighting observations rather than altering covariates.  
By contrast, SCDA modifies the covariate distribution directly, producing an estimator that can be interpreted as implicitly down-weighting the information carried by spurious features.  
Section~\ref{sec:efficiency} makes these relationships precise by comparing the asymptotic variances of SCDA, ERM, and IPW estimators.  

\subsection{Functionals of the SCDA Estimator}

While our asymptotic results are formulated in terms of the parameter 
$\hat{\theta}_{\mathrm{SCDA}}$, practical interest often centers on functionals 
$\mu = g(\theta_0)$ for smooth maps $g: \Theta \to \mathbb{R}^q$. Examples include groupwise prediction accuracies $\mu_g = \mathbb{P}(f_{\theta_0}(X) = Y \mid G = g)$ 
for subgroups defined by nuisance attributes
or risk measures such as worst-group loss or fairness metrics.
By the delta method, if $g$ is continuously differentiable at $\theta_0$, then
\[
\sqrt{n}\left(g(\hat{\theta}_{\mathrm{SCDA}}) - g(\theta_0)\right) 
\;\overset{d}{\longrightarrow}\; 
\mathcal{N}\left(0, \nabla g(\theta_0)^\top \Sigma \, \nabla g(\theta_0)\right),
\]
where $\Sigma$ is given in Theorem~\ref{thm:clt}. This result enables inference 
on any smooth functional of interest, with variance consistently estimated by 
$\nabla g(\hat{\theta})^\top \hat{\Sigma} \, \nabla g(\hat{\theta})$.

In our empirical studies (Section~\ref{sec:simulations}), we evaluate both
coefficient estimates in parametric models and predictive functionals such as
the model prediction at a covariate point, all of which fall within this framework.

\subsection{Practical Guidance for Constructing Valid Augmentation Mechanisms}
\label{sec:practical_guidance}

The augmentation $\mathcal{A}$ must preserve the feature of interest $X_1$ while rendering the
nuisance $X_2$ uninformative about $Y$ (Assumption~\ref{ass:aug}). Its construction rests on
domain knowledge of which features are of interest and which are nuisance: anatomy versus imaging
artifacts in medical imaging, semantic content versus stylistic markers in text. Three general
strategies are available: \emph{permutation}, which reassigns $X_2$ across observations (simple, no
modeling, but may create implausible combinations); \emph{conditional generation}, which samples
$\tilde X_2 \sim \hat p(X_2\mid X_1)$ (more realistic, at the cost of modeling error); and
\emph{masking}, which sets $\tilde X_2$ to a null value (simple, but may induce support shift).

Formally, write $\mathcal{A}(x) = \delta_{x_1} \otimes q(\cdot \mid x)$ for the
conditional law of $\tilde X = (x_1, \tilde X_2)$. The three strategies correspond to explicit
conditional densities: permutation draws $\tilde X_2$ from the marginal,
$q(\tilde x_2 \mid x) = p_{X_2}(\tilde x_2)$ (empirically, the uniform law on the observed
$\{X_{j2}\}_{j=1}^n$); conditional generation draws from a fitted conditional,
$q(\tilde x_2 \mid x) = \hat p(\tilde x_2 \mid x_1)$; and masking sets
$q(\cdot \mid x) = \delta_{x_2^{0}}$ for a fixed null value $x_2^{0}$. As a concrete example, in the
Gaussian model of Section~\ref{sec:simulations} ($Y = X_1 + \varepsilon$, $X_2 = \rho Y + \eta$ with
standard normal errors), permutation yields $\tilde X_2 \sim N(0,\, 2\rho^2 + 1)$ independent of
$(X_1, Y)$, so $\mathcal{A}(x) = \delta_{x_1} \otimes N(0,\, 2\rho^2+1)$ in closed form and
Assumption~\ref{ass:aug} holds exactly. None of these mechanisms depends on $\theta$;
$\theta$-dependence arises only for adaptive schemes in which the augmentation is indexed by the
current estimate (Section~\ref{sec:iterative}).

Before applying SCDA, practitioners should check that $\mathcal{A}$ approximately satisfies
Assumption~\ref{ass:aug}: (i) \emph{feature preservation}, that $\tilde X_1 = X_1$ (the requirement of
Assumption~\ref{ass:aug}(i); e.g.\ perceptual similarity on the regions of interest, or
$\|\tilde X_{i1}-X_{i1}\|$ small); and (ii) \emph{spurious
attenuation}, that the conditional association $\mathrm{Cor}(\tilde X_2, Y\mid X_1)$ is substantially
smaller than $\mathrm{Cor}(X_2,Y\mid X_1)$, using a conditional-independence test where feasible
\citep{zhang2012kernel}. These checks require that $X_1$ be observable, as in the
separable case or when an annotated validation sample is available; in the entangled case they are
unavailable, and the sensitivity analysis of Section~\ref{sec:sensitivity} is the appropriate
recourse. If latent confounding renders exact independence unattainable, SCDA remains
useful under the weaker requirement that the augmentation \emph{reduces} the spurious association;
we recommend sensitivity analysis across augmentation strategies. Detailed diagnostics, thresholds,
and guidance for handling detected violations (which, when rare in the sense
$\epsilon_n=o(n^{-1/2})$ of Proposition~\ref{prop:approx_preserve}, leave the asymptotic
distribution intact) are given in Supplementary Section~\ref{S-sec:practical_supp}.

\section{Asymptotic Properties}
\label{sec:theory}

We now establish the large-sample properties of the statistical counterfactual data augmentation (SCDA) estimator $\hat\theta_{\mathrm{SCDA}}$. Our analysis builds on classical M-estimation theory \citep{vaart1998asymptotic,van2000asymptotic}. The inner expectation over augmentations places SCDA within the framework of simulated-moment / simulated-score estimation \citep{mcfadden1989method,pakes1989simulation}: the estimating equation is a moment condition evaluated by simulation from $\mathcal{A}$, and the finite-draw variance inflation we derive in Section~\ref{sec:montecarlo} is the counterpart of the simulation-noise term in that literature. What is specific to SCDA is that the simulation distribution $\mathcal{A}$ is not a computational device for an intractable integral but a modeling choice that \emph{defines} the estimand (Lemma~\ref{lem:identity}); our contribution is to make the resulting inference, including the finite-draw correction, usable for counterfactual augmentation, and to characterize its efficiency (Section~\ref{sec:efficiency}).

\subsection{Assumptions}

Let 
\[
\Psi(\theta) = \mathbb{E}_{(X,Y)\sim P_0}\,\mathbb{E}_{\tilde X \sim \mathcal{A}(X)} \nabla_\theta \ell(Y,f_\theta(\tilde X))
\]
denote the population estimating function. We define $\theta_0 \in \Theta$ as a solution to $\Psi(\theta_0)=0$. Under standard conditions (e.g., strict convexity of the loss in $\theta$, or when $\Psi(\theta)$ is injective), uniqueness is expected to hold. 

We impose the following regularity conditions for the theoretical analysis.

\begin{assumption}[Loss and Model Regularity]\label{ass:loss}
The loss function $\ell(y,f_\theta(x))$ is measurable in $(y,x)$, twice continuously differentiable in $\theta$, and satisfies
\[
\sup_{\theta \in \Theta}\, \mathbb{E}\,\|\nabla_\theta \ell(Y,f_\theta(X))\|^2 < \infty.
\]
\end{assumption}

\begin{assumption}[Augmentation Mechanism]\label{ass:aug}
For each $x = (x_1, x_2) \in \mathcal{X}$ and $\theta \in \Theta$:
\begin{enumerate}[(i)]
\item The augmentation distribution $\mathcal{A}(x)$ is supported on 
$\{(x_1, \tilde{x}_2): \tilde{x}_2 \in \mathcal{X}_2\}$, i.e., 
$\mathbb{P}(\tilde{X}_1 = x_1 \mid X = x) = 1$.

\item The map $(x,y,\theta) \mapsto \mathbb{E}_{\tilde{X} \sim \mathcal{A}(x)}
[\nabla_\theta \ell(y, f_\theta(\tilde{X}))]$ is continuously differentiable in $\theta$
for $P_0$-almost every $(x,y)$.

\item There exist integrable functions $G, H: \mathcal{X} \times \mathcal{Y} \to \mathbb{R}_+$ 
such that for all $\theta \in \Theta$ and $(x,y)$,
\begin{align*}
\sup_{\tilde{x} \in \text{supp}(\mathcal{A}(x))} 
\|\nabla_\theta \ell(y, f_\theta(\tilde{x}))\| &\leq G(x,y), \\
\sup_{\tilde{x} \in \text{supp}(\mathcal{A}(x))} 
\|\nabla^2_\theta \ell(y, f_\theta(\tilde{x}))\| &\leq H(x,y),
\end{align*}
with $\mathbb{E}_{P_0}[G(X,Y)^2] < \infty$ and $\mathbb{E}_{P_0}[H(X,Y)] < \infty$.
\end{enumerate}
\end{assumption}

Intuitive explanations and practical guidance for Assumption~\ref{ass:aug} are provided in Section 2.5. Condition (i) ensures that the features of interest $X_1$ are preserved during augmentation, so that the augmented data retain the $X_1$--$Y$ relationship that defines the target. Condition (ii) ensures that the expected gradient is smooth, which is needed for optimization and asymptotic analysis. Condition (iii) ensures that the augmentation does not introduce unbounded variability, which is necessary for establishing consistency and asymptotic normality.

\begin{remark}[On $\theta$-dependence and smoothness]
In the one-shot estimator studied in Sections~\ref{sec:method}--\ref{sec:efficiency}, the
augmentation distribution $\mathcal{A}$ is fixed and does not depend on $\theta$: the common
mechanisms (permutation of nuisance attributes, background swapping in images, conditional
generation, or masking) are all constructed independently of the parameter being estimated.
The adaptive and iterative schemes in which the augmentation evolves with the current
estimator are treated separately in Section~\ref{sec:iterative}.

Assumption~\ref{ass:aug}(ii) requires continuous differentiability of the composite
map $f_\theta$ under the augmentation. While Assumption~\ref{ass:loss}
ensures $\ell(y, f_\theta(x))$ is twice differentiable in $\theta$ for fixed $(y,x)$,
we additionally require that differentiation can be exchanged with the expectation
over the augmentation distribution $\mathcal{A}$. This is guaranteed by
the dominated convergence conditions in Assumption~\ref{ass:aug}(iii), which are
essential for establishing the asymptotic normality and variance estimation results.
\end{remark}
\begin{assumption}[Compactness]\label{ass:compact}
The parameter space $\Theta$ is compact, and the support of $\mathcal{A}(x)$ 
is contained in a $P_0$-measurable set for all $\theta \in \Theta$ and $x \in \mathcal{X}$.
\end{assumption}

\begin{assumption}[Identifiability]\label{ass:ident}
The population risk $R(\theta)$ has a unique minimizer $\theta_0 \in \mathrm{int}(\Theta)$ that is
\emph{well separated}: for every $\delta>0$,
$\inf_{\theta:\|\theta-\theta_0\|\ge\delta} R(\theta) > R(\theta_0)$. Since $\theta_0$ is an interior
minimizer and $R$ is differentiable (Assumptions~\ref{ass:loss}--\ref{ass:aug}), it follows that
$\Psi(\theta_0)=\nabla_\theta R(\theta_0)=0$, and $\theta_0$ is the unique stationary point in a
neighborhood of itself.
\end{assumption}

\begin{assumption}[Non-Singularity]\label{ass:nonsing}
The Jacobian matrix $A_0 = \nabla_\theta \Psi(\theta)\vert_{\theta=\theta_0}$ exists and is nonsingular.
\end{assumption}

These assumptions are mild generalizations of those used in standard M-estimation. The main novelty is Assumption~\ref{ass:aug}, which ensures that the augmentation mechanism preserves the features of interest while attenuating spurious associations.

\noindent\textbf{Scope and nonconvex estimators.}
The asymptotic results in this section are stated for estimators that solve 
the estimating equation $\Psi_n(\theta)=0$ under Assumptions~\ref{ass:loss}--\ref{ass:nonsing}. 
These conditions are standard in parametric and smooth semiparametric models. 
For nonconvex estimators such as deep neural networks, global minimizers may 
be difficult to characterize; however, the results continue to apply to any 
sequence of estimators that converges to a local minimizer satisfying the 
stated regularity conditions. The compactness assumption (Assumption~\ref{ass:compact}) 
can be relaxed to local compactness in a neighborhood of $\theta_0$ for such applications. 
Extensions to overparameterized and nonconvex settings, where the effective parameter 
space depends on $n$, require additional analysis beyond the scope of this work.

\subsection{Consistency}

\begin{theorem}[Consistency]\label{thm:consistency}
Under Assumptions \ref{ass:loss}--\ref{ass:nonsing}, the SCDA estimator $\hat\theta_{\mathrm{SCDA}}$ converges in probability to $\theta_0$, i.e.
\[
\hat\theta_{\mathrm{SCDA}} \xrightarrow{p} \theta_0, \qquad n \to \infty.
\]
\end{theorem}

\begin{proof}[Sketch of proof]
By Assumptions \ref{ass:loss} and \ref{ass:aug}, the class of functions 
\[
(x,y) \mapsto \mathbb{E}_{\tilde X \sim \mathcal{A}(x)} \ell(y,f_\theta(\tilde X))
\]
satisfies a uniform law of large numbers. Consequently, $R_n(\theta)$ converges uniformly to $R(\theta)$, and minimizers of $R_n(\theta)$ converge to the unique minimizer of $R(\theta)$, namely $\theta_0$. The complete proof is given in Supplementary Section~\ref{S-pf:consistency}.
\end{proof}

\subsection{Asymptotic Normality}
Before stating the theorem, we define the Jacobian matrix $A_0 = \nabla_\theta \Psi(\theta)|_{\theta=\theta_0}$, where $\Psi(\theta)$ is the population estimating function defined in Section~\ref{sec:theory}.

\begin{theorem}[Asymptotic Normality of SCDA]\label{thm:clt}
Under Assumptions \ref{ass:loss}--\ref{ass:nonsing}, 
\[
\sqrt{n}\,(\hat\theta_{\mathrm{SCDA}}-\theta_0) \;\overset{d}{\longrightarrow}\; \mathcal{N}(0,\Sigma),
\]
where the asymptotic covariance has the sandwich form
\[
\Sigma = A_0^{-1} B_0 (A_0^{-1})^\top, \qquad
B_0 = \operatorname{Var}\!\left( \mathbb{E}_{\tilde X \sim \mathcal{A}(X)} \nabla_\theta \ell(Y,f_{\theta_0}(\tilde X)) \right),
\]
and $A_0 = \nabla_\theta \Psi(\theta)|_{\theta=\theta_0}$ is the Jacobian matrix defined above.
\end{theorem}

The covariance structure differs fundamentally from that of ERM and IPW: in SCDA the variance contribution comes from \emph{expectations over counterfactuals}, which attenuates variability driven by spurious features.

\begin{proof}[Sketch of proof]
The estimating equation $\Psi_n(\theta)=0$ defines an M-estimator with score contributions 
\[
\psi(Z_i;\theta)=\mathbb{E}_{\tilde X \sim \mathcal{A}(X_i)} \nabla_\theta \ell(Y_i,f_\theta(\tilde X)), 
\quad Z_i=(X_i,Y_i).
\] 
By Assumptions \ref{ass:loss}--\ref{ass:nonsing}, a Taylor expansion of $\Psi_n(\hat\theta_{\mathrm{SCDA}})$ around $\theta_0$ yields
\[
0=\Psi_n(\hat\theta_{\mathrm{SCDA}}) = \Psi_n(\theta_0) + A_0(\hat\theta_{\mathrm{SCDA}}-\theta_0)+o_p(\|\hat\theta_{\mathrm{SCDA}}-\theta_0\|).
\]
The central limit theorem applies to the i.i.d. mean $\Psi_n(\theta_0)$ with variance $B_0$. Rearranging establishes the stated distributional limit. The complete proof is given in Supplementary Section~\ref{S-pf:clt}.
\end{proof}
\subsection{Robustness to Imperfect Feature Preservation}

In practice, augmentation procedures may occasionally perturb components of the feature of interest $X_1$ due to imperfect segmentation, generator error, or approximation. The next proposition quantifies the effect of such rare perturbations on the SCDA estimator.

\begin{assumption}[Approximate Preservation]\label{ass:approx_preserve}
The augmentation mechanism $\mathcal{A}$ satisfies the following: For each draw $\tilde{X} \sim \mathcal{A}(X)$, there exists a (possibly $n$-dependent) probability $\epsilon_n \in [0,1)$ such that
\begin{enumerate}[(i)]
\item With probability $1-\epsilon_n$, the augmentation perfectly preserves the feature of interest $X_1$, i.e., $\tilde{X}_1 = X_1$ almost surely.
\item With probability $\epsilon_n$, the augmentation may perturb $X_1$ (resulting in $\tilde{X}_1 \neq X_1$), but the perturbed covariate $\tilde{X} = (\tilde{X}_1, \tilde{X}_2)$ remains in the support of the original distribution. Moreover, when conditioning on the event that $\tilde{X}_1 \neq X_1$, the following bound holds:
\[
\mathbb{E}_{\tilde{X} \sim \mathcal{A}(X)}\left[\|\nabla_\theta \ell(Y, f_\theta(\tilde{X})) - \nabla_\theta \ell(Y, f_\theta(X_1, \tilde{X}_2))\|^2 \,\Big|\, \tilde{X}_1 \neq X_1\right] \leq M^2,
\]
where $X_1$ on the right-hand side refers to the original (unperturbed) feature of interest, $(X_1, \tilde{X}_2)$ denotes the covariate with original $X_1$ and perturbed $X_2$, and $M < \infty$ is a constant that does not depend on $n$.
\end{enumerate}
We assume $\epsilon_n \to 0$ as $n \to \infty$.
\end{assumption}
\begin{proposition}[Robustness to Rare Perturbations of $X_1$]\label{prop:approx_preserve}
Suppose Assumptions \ref{ass:loss}--\ref{ass:nonsing} and \ref{ass:approx_preserve} hold, with
perturbation probability $\epsilon_n = o(1)$ and bound $M<\infty$ as defined there. Then:
\begin{enumerate}[(a)]
\item The population target under the imperfect mechanism, $\theta_{\epsilon_n}$
(the solution of the perturbed estimating equation), satisfies
$\|\theta_{\epsilon_n} - \theta_0\| = O(\epsilon_n)$.
\item If $n\epsilon_n^2 \to 0$, then 
$$\sqrt{n}(\hat{\theta}_{\mathrm{SCDA}} - \theta_0) \overset{d}{\longrightarrow} 
\mathcal{N}(0, \Sigma).$$
\end{enumerate}
\end{proposition}

\begin{proof}
See Supplementary Section~\ref{S-pf:robust}.
\end{proof}
\begin{remark}
The condition $n\epsilon_n^2 \to 0$ in part (b) is equivalent to $\epsilon_n = o(n^{-1/2})$ 
up to logarithmic factors. This condition ensures that the bias from perturbations of $X_1$ 
shrinks faster than the stochastic error $O_p(n^{-1/2})$, so the latter dominates in the 
asymptotic distribution. The bounded second moment condition in 
Assumption~\ref{ass:approx_preserve} ensures that occasional perturbations do not introduce 
heavy-tailed score contributions that would invalidate the CLT.
\end{remark}

\subsection{Variance Estimation}

Practical inference requires consistent estimation of $\Sigma = A_0^{-1}B_0(A_0^{-1})^\top$. 
When the inner expectation over $\mathcal{A}$ can be computed exactly or approximated 
via Monte Carlo (Section~\ref{sec:montecarlo}), standard plug-in estimators apply.

Recall that $A_0 = \nabla_\theta \Psi(\theta)|_{\theta=\theta_0}$ where 
$\Psi(\theta) = \mathbb{E}[\mathbb{E}_{\tilde{X} \sim \mathcal{A}(X)} 
\nabla_\theta \ell(Y, f_\theta(\tilde{X}))]$. Define the empirical Jacobian
\[
\hat{A} = \frac{1}{n}\sum_{i=1}^n \mathbb{E}_{\tilde{X} \sim \mathcal{A}(X_i)} 
\nabla^2_\theta \ell(Y_i, f_{\hat{\theta}}(\tilde{X})),
\]
and the empirical variance of scores
\[
\hat{B} = \frac{1}{n}\sum_{i=1}^n (\hat{\psi}_i - \bar{\psi}_n)(\hat{\psi}_i - \bar{\psi}_n)^\top,
\]
where 
\[
\hat{\psi}_i = \mathbb{E}_{\tilde{X} \sim \mathcal{A}(X_i)} 
\nabla_\theta \ell(Y_i, f_{\hat{\theta}}(\tilde{X})), \quad 
\bar{\psi}_n = \frac{1}{n}\sum_{i=1}^n \hat{\psi}_i.
\]

\begin{proposition}[Variance estimation consistency]\label{prop:var_est}
Under Assumptions \ref{ass:loss}--\ref{ass:nonsing}, $\hat{A} \xrightarrow{p} A_0$ 
and $\hat{B} \xrightarrow{p} B_0$. Consequently, 
$\hat{\Sigma} = \hat{A}^{-1}\hat{B}(\hat{A}^{-1})^\top \xrightarrow{p} \Sigma$.
\end{proposition}

\begin{proof}
See Supplementary Section~\ref{S-pf:var}.
\end{proof}
\begin{remark}
In practice, the Hessian matrices in $\hat{A}$ may be computed via automatic 
differentiation. Alternatively, when the Hessian is computationally prohibitive, 
practitioners may approximate $\hat{A}$ using the outer product of gradients or 
numerical finite differences. Numerical finite differencing of $\Psi_n$ yields a consistent estimate of $A_0$ under the same conditions. The outer-product-of-gradients approximation, by contrast, targets $B_0$ rather than $A_0$ and coincides with it only under a generalized information equality \citep{white1982maximum}, which need not hold for augmented risks; we therefore recommend automatic differentiation or finite differences for $\hat A$.
\end{remark}

\textbf{Monte Carlo implementation.} When expectations over $\mathcal{A}$ are 
approximated using $m$ draws per observation (Section~\ref{sec:montecarlo}), replace 
$\hat{\psi}_i$ with its Monte Carlo analogue $\hat{\psi}_{i,m} = \frac{1}{m}\sum_{j=1}^m 
\nabla_\theta \ell(Y_i, f_{\hat{\theta}}(\tilde{X}_{ij}))$. If $m$ is fixed, the 
variance estimator should include the additional $\hat{V}_{\text{MC}}/m$ term from 
Proposition~\ref{prop:mcvar}, as detailed in Section~\ref{sec:montecarlo}.

\subsection{Monte Carlo Approximation of the Inner Expectation}
\label{sec:montecarlo}

Since the inner expectation over $\tilde{X} \sim \mathcal{A}(X)$ is generally 
intractable, practical implementations approximate it via Monte Carlo sampling. 
For each observation $(X_i,Y_i)$, we generate $m$ independent augmented covariates 
$\tilde{X}_{i1},\ldots,\tilde{X}_{im} \stackrel{\text{iid}}{\sim} \mathcal{A}(X_i)$ 
and replace the exact expectation with an empirical average. This yields the 
Monte Carlo empirical risk
\[
R_{n,m}(\theta) = \frac{1}{n}\sum_{i=1}^n \frac{1}{m}\sum_{j=1}^m 
\ell(Y_i, f_\theta(\tilde{X}_{ij})),
\]
with corresponding estimator
\[
\hat{\theta}_{n,m} = \arg\min_{\theta \in \Theta} R_{n,m}(\theta).
\]

\begin{assumption}[Finite augmentation draws]\label{ass:mc}
For each $x \in \mathcal{X}$ and $\theta \in \Theta$, the augmented covariates 
$\tilde{X}_1,\ldots,\tilde{X}_m \sim \mathcal{A}(x)$ are i.i.d.\ given $X=x$. 
Moreover,
\[
\sup_{\theta \in \Theta}\,\mathbb{E}\,\|\nabla_\theta \ell(Y, f_\theta(\tilde{X}))\|^2 
< \infty
\]
uniformly over draws from $\mathcal{A}(X)$.
\end{assumption}

The key insight is that Monte Carlo approximation introduces an additional source 
of variability beyond sampling variability in $(X_i,Y_i)$. To quantify this effect, 
decompose the total variance into between-observation and within-observation components.

\begin{proposition}[Asymptotic distribution with finite augmentation]
\label{prop:mcvar}
Suppose Assumptions \ref{ass:loss}--\ref{ass:nonsing} and \ref{ass:mc} hold. 
Then the Monte Carlo SCDA estimator satisfies
\[
\sqrt{n}(\hat{\theta}_{n,m} - \theta_0) \;\overset{d}{\longrightarrow}\; 
\mathcal{N}(0, \Sigma_m),
\]
where
\[
\Sigma_m = A_0^{-1}\left(B_0 + \frac{1}{m}V_{\text{MC}}\right)(A_0^{-1})^\top,
\]
with $B_0$ as in Theorem~\ref{thm:clt} and
\[
V_{\text{MC}} = \mathbb{E}\left[\operatorname{Var}\left(\nabla_\theta 
\ell(Y, f_{\theta_0}(\tilde{X})) \,\big|\, X, Y\right)\right].
\]
\end{proposition}

\begin{proof}
See Supplementary Section~\ref{S-pf:mc}.
\end{proof}
\textbf{Implications for practice.} Proposition~\ref{prop:mcvar} has three important 
consequences:

\begin{enumerate}
\item If $m \to \infty$ as $n \to \infty$, the Monte Carlo error vanishes asymptotically 
and $\hat{\theta}_{n,m}$ achieves the same limiting distribution as the exact SCDA 
estimator in Theorem~\ref{thm:clt}.

\item If $m$ is held fixed (common in deep learning), the estimator remains consistent 
and asymptotically normal but with inflated variance. The inflation term
$A_0^{-1}V_{\text{MC}}A_0^{-\top}/m$ (in the scalar case, a factor $1 + V_{\text{MC}}/(m B_0)$)
quantifies the efficiency loss from finite augmentation.

\item For valid inference with fixed $m$, practitioners must account for the 
$V_{\text{MC}}/m$ term in variance estimation. This is done by using the variance estimator $\hat{\Sigma}_m$ defined in equation~\eqref{eq:sigma_m} below, which includes the $\hat{V}_{\text{MC}}/m$ term. Practitioners should use $\hat{\Sigma}_m$ (not $\hat{\Sigma}$) when $m$ is finite.
\end{enumerate}

\textbf{Variance estimation with Monte Carlo approximation.} 
The asymptotic covariance $\Sigma_m$ can be consistently estimated via
\begin{equation}\label{eq:sigma_m}
\hat{\Sigma}_m = \hat{A}^{-1}\left(\hat{B} + \frac{1}{m}\hat{V}_{\text{MC}}\right)
(\hat{A}^{-1})^\top,
\end{equation}
where $\hat{A}$ is as in Section~\ref{sec:theory} and $\hat{B}$, $\hat{V}_{\text{MC}}$ are the
empirical between- and within-observation score covariances; explicit formulas and practical
guidance on choosing $m$ (based on the ratio $\hat{V}_{\text{MC}}/(m\hat{B})$) are given in
Supplementary Section~\ref{S-sec:practical_supp}. The correction matters only for finite $m$ and for
functionals loading on the augmented feature; for the target parameter it is asymptotically
negligible (Theorem~\ref{thm:oracle}).

\section{Efficiency and Robustness}\label{sec:efficiency}
\subsection{Efficiency Relative to Empirical Risk Minimization}

The asymptotic covariance of the SCDA estimator in Theorem~\ref{thm:clt} is 
\[
\Sigma_{\mathrm{SCDA}} = A_0^{-1} B_0 (A_0^{-1})^\top,
\]
where 
\[
B_0 = \operatorname{Var}\!\left( \mathbb{E}_{\tilde X \sim \mathcal{A}(X)} 
\nabla_\theta \ell(Y,f_{\theta_0}(\tilde X)) \right).
\]
To interpret this expression, it is instructive to compare with the empirical risk minimizer (ERM), which solves
\[
\Psi^{\mathrm{ERM}}_n(\theta) = \frac{1}{n}\sum_{i=1}^n \nabla_\theta \ell(Y_i,f_\theta(X_i))=0,
\]
with population score $\Psi^{\mathrm{ERM}}(\theta) = \mathbb{E}\{\nabla_\theta \ell(Y,f_\theta(X))\}$.  
The ERM estimator is asymptotically normal with covariance
\[
\Sigma_{\mathrm{ERM}} = (A^{\mathrm{ERM}}_0)^{-1} B^{\mathrm{ERM}}_0 (A^{\mathrm{ERM}}_0)^{-\top},
\]
where $B^{\mathrm{ERM}}_0 = \operatorname{Var}\!\big(\nabla_\theta \ell(Y,f_{\theta_0}(X))\big)$.  

A Loewner comparison of $\Sigma_{\mathrm{SCDA}}$ and $\Sigma_{\mathrm{ERM}}$ is not the
right one: ERM is biased for the target defined below, so a variance ordering conflates bias and
variance. The proper comparison is in terms of \emph{risk under the target distribution}, which we
now make precise.

Decompose $X=(X_1,X_2)$ into features of primary interest and nuisance features.
The augmentation \emph{defines the estimand}. Let $P_{\mathcal{A}}$ denote the joint law of
$(\tilde X, Y)$ when $(X,Y)\sim P_0$ and $\tilde X \sim \mathcal{A}(X)$; for permutation augmentation,
which draws $\tilde X_2$ from its marginal independently of $(X_1,Y)$ while preserving $X_1$,
$P_{\mathcal{A}} = P_{X_1,Y}\otimes P_{X_2}$.

\begin{lemma}[Augmented risk is target risk]\label{lem:identity}
For every $\theta\in\Theta$, $R(\theta) = \mathbb{E}_{P_{\mathcal{A}}}[\ell(Y,f_\theta(X))]$.
Consequently $\theta_0 = \arg\min_{\theta\in\Theta}\mathbb{E}_{P_{\mathcal{A}}}[\ell(Y,f_\theta(X))]$;
that is, the SCDA estimand is the risk minimizer under the nuisance-marginalized target
$P_{\mathcal{A}}$.
\end{lemma}
\begin{proof}
See Supplementary Section~\ref{S-pf:eff}.
\end{proof}

The efficiency comparison is therefore in terms of excess risk under $P_{\mathcal{A}}$,
\[
\mathcal{E}(\theta) \;=\; \mathbb{E}_{P_{\mathcal{A}}}[\ell(Y,f_\theta(X))]
\;-\; \min_{\theta'\in\Theta}\mathbb{E}_{P_{\mathcal{A}}}[\ell(Y,f_{\theta'}(X))] \;\ge\; 0.
\]

\begin{theorem}[Excess risk relative to ERM]\label{thm:efficiency}
Suppose Assumptions~\ref{ass:loss}--\ref{ass:nonsing} hold and
$H = \nabla^2_\theta \mathbb{E}_{P_{\mathcal{A}}}[\ell(Y,f_\theta(X))]\big|_{\theta_0}\succ 0$. Then:
\begin{enumerate}[(a)]
\item $n\,\mathcal{E}(\hat\theta_{\mathrm{SCDA}}) \xrightarrow{d} \tfrac12 W^\top H W$ with
$W\sim\mathcal{N}(0,\Sigma)$, so $\mathcal{E}(\hat\theta_{\mathrm{SCDA}}) = O_p(n^{-1})$;
\item if $\theta_{\mathrm{ERM}}\neq\theta_0$, then
$\mathcal{E}(\hat\theta_{\mathrm{ERM}}) \xrightarrow{p} \Delta > 0$, a constant independent of $n$;
\item consequently
$\mathbb{P}\big(\mathcal{E}(\hat\theta_{\mathrm{SCDA}}) < \mathcal{E}(\hat\theta_{\mathrm{ERM}})\big)\to 1$.
\end{enumerate}
\end{theorem}
\begin{proof}
See Supplementary Section~\ref{S-pf:eff}.
\end{proof}

Thus SCDA's risk under the target vanishes at the parametric rate, while ERM's remains
bounded away from optimal by a fixed gap $\Delta$. In the Gaussian location model of
Section~\ref{sec:simulations}, $\Delta = \rho^2(3\rho^2+1)/(1+\rho^2)^2$ in closed form
(Supplementary Section~\ref{S-pf:eff}).

The excess-risk comparison against ERM concerns \emph{bias}. We now turn to
\emph{variance}, and establish a stronger positive property: SCDA is oracle-efficient for the
target and its efficiency does not depend on which valid augmentation is used.

\begin{theorem}[Oracle efficiency and augmentation-invariance]\label{thm:oracle}
Write the score in the canonical form $\nabla_\theta \ell(Y,f_{\theta_0}(X)) =
r(X,Y)\,\nabla_\theta f_{\theta_0}(X)$, where $r(X,Y) = \partial \ell/\partial \eta$ evaluated at
$\eta=f_{\theta_0}(X)$ is the outcome residual (e.g.\ $r=-2(Y-f_{\theta_0})$ for squared loss,
$r=f_{\theta_0}-Y$ for logistic loss). Suppose Assumptions~\ref{ass:loss}--\ref{ass:nonsing} hold
and the model is correctly specified on $X_1$, in the sense that $\mathbb{E}[r(X,Y)\mid X_1]=0$ at
$\theta_0$. Let $\mathcal{Q}$ be the class of valid augmentation kernels (preserving $X_1$ and
rendering $\tilde X_2\perp Y\mid X_1$). Then, for the target parameter $\theta_1$:
\begin{enumerate}[(a)]
\item \textbf{(Invariance.)} The asymptotic variance of
$\hat\theta_{1,\mathrm{SCDA}}$ is the same for every $\mathcal{A}\in\mathcal{Q}$.
\item \textbf{(Oracle efficiency.)} That common variance equals the asymptotic
variance of the oracle estimator that regresses $Y$ on $X_1$ alone; SCDA attains oracle
efficiency without the analyst identifying or isolating the features of interest.
\item \textbf{(Finite-$m$ refinement.)} With $m$ Monte Carlo draws per
observation, the only kernel-dependent term is the inflation
$V_{\mathrm{MC}}(\mathcal{A})/m$ of Proposition~\ref{prop:mcvar}, which is nonnegative, equals
zero for deterministic kernels, and is otherwise strictly positive.
\end{enumerate}
\end{theorem}
\begin{proof}
See Supplementary Section~\ref{S-pf:oracle}.
\end{proof}

The practical reading of Theorem~\ref{thm:oracle} is that the choice among valid
augmentations (permutation, conditional generative, masking) is asymptotically irrelevant for
the target and should be governed by finite-$m$ cost and by validity/support considerations
(Section~\ref{sec:practical_guidance}) rather than by asymptotic efficiency.

\subsection{Comparison with Inverse Probability Weighting}

Inverse probability weighting (IPW) \citep{horvitz1952generalization,rosenbaum1983central} offers an
alternative that reweights samples, with score
$\Psi^{\mathrm{IPW}}_n(\theta) = n^{-1}\sum_i w(X_i)\nabla_\theta\ell(Y_i,f_\theta(X_i))$ and weights
$w(X_i)=p(X_{i1})/p(X_{i1},X_{i2})$ estimated from data. SCDA differs in two ways. First, it targets a
different object, the risk minimizer under $P_{\mathcal{A}}$ (Lemma~\ref{lem:identity}), rather than
the reweighted population, so we do not assert a covariance ordering. Second, the IPW score multiplies
the gradient by a random weight $w(X)$, whose variability (compounded by estimation error) inflates
its sandwich meat, whereas the SCDA score is an expectation with no random weight. The qualitative
consequence is that SCDA achieves bias correction by modifying covariates rather than reweighting,
avoiding the weight-induced instability that affects IPW under strong spurious association. Like model
selection and doubly-robust estimation, IPW is a functional of the observed law and cannot reach the
target in the entangled regime of Section~\ref{sec:sim_entangled}.
\subsection{Robustness and Influence Function Analysis}

The SCDA estimator has influence function $\mathrm{IF}(Z) = -A_0^{-1}\psi(Z;\theta_0)$ with
$\psi(Z;\theta)=\mathbb{E}_{\tilde X\sim\mathcal{A}(X)}\nabla_\theta\ell(Y,f_\theta(\tilde X))$, the
standard M-estimator form. When occasional augmentation errors (perturbing $X_1$, or leaving residual
$X_2$--$Y$ dependence) contribute a bounded increment to the score, $\|\mathrm{IF}(Z)\|$ stays
uniformly bounded, giving SCDA \emph{B-robustness} \citep{huber1981robust,hampel1986robust}: isolated
augmentation errors add negligible noise without destabilizing the estimator. Proposition~\ref{prop:approx_preserve}
makes this precise for rare perturbations of $X_1$; Section~\ref{sec:sensitivity} treats systematic
invalidity, where the bounded-increment condition fails, by sensitivity analysis rather than
robustness.

\section{Sensitivity Analysis for Augmentation Validity}
\label{sec:sensitivity}
\begingroup

The results so far assume the augmentation is \emph{valid}: the resampled nuisance satisfies
$\tilde X_2 \perp Y \mid X_1$ (Assumption~\ref{ass:aug}), so that treating $X_2$ as pure nuisance is
correct. This assumption is not testable when $X_1$ and $X_2$ are entangled in the observed
representation (the regime of Section~\ref{sec:sim_entangled}); Proposition~\ref{prop:untestable}
below makes this precise, so the residual dependence is not identified. Rather than impose it, we ask how far
the conclusions can move if validity fails, and answer with sharp bounds and a breakdown frontier
obtained by indexing departures from validity by a single parameter $\Gamma$, in the spirit of
sensitivity analysis for untestable assumptions \citep{rosenbaum2002observational,manski2003partial}.
The point-identified theory of Sections~\ref{sec:theory}--\ref{sec:efficiency} is recovered as the
$\Gamma=0$ special case.

\subsection{A sensitivity model}

Validity requires that, in the augmented data, the resampled nuisance carries no residual
information about the outcome beyond $X_1$. We index departures from validity by the residual
partial dependence.

\begin{definition}[$\Gamma$-valid augmentation]\label{def:gammavalid}
For $\Gamma\ge 0$, an augmentation mechanism $\mathcal{A}$ is \emph{$\Gamma$-valid} if it preserves
$X_1$, matches the identified conditional law of $X_2$ given $X_1$ in its first two moments, and its
residual partial correlation with the outcome is bounded,
\[
|\gamma(\mathcal{A})| \;:=\; \bigl|\operatorname{Corr}\!\big(\tilde X_2,\,Y \,\big|\, X_1\big)\bigr|
\;\le\; \Gamma .
\]
$\Gamma=0$ is exact validity (Assumption~\ref{ass:aug}); $\Gamma$ bounds the residual
nuisance--outcome correlation that the augmentation fails to remove.
\end{definition}

The scale $\Gamma$ is interpretable and, in principle, boundable from domain knowledge or from an
auxiliary validation sample: it is the largest conditional correlation between the resampled nuisance
and the outcome that the analyst is willing to entertain. We take the identified conditional law of
$X_2$ given $X_1$ as known (it is estimable from the observed data); the only unidentified freedom is
the residual outcome dependence, which $\Gamma$ constrains.

The sensitivity model is not a conservative add-on: in the entangled regime the residual dependence
$\gamma$ is genuinely unidentified, so no procedure can recover it from the data.

\begin{proposition}[Untestability of validity]\label{prop:untestable}
Consider the entangled setting in which the analyst observes $Z=MX$ with $M$ unknown and invertible,
together with the augmented input $\tilde Z$ and the outcome $Y$. There exist two latent
representations that generate the identical joint law of $(Z,\tilde Z,Y)$ but carry different
residual dependence: $\gamma=0$ under one and $\gamma>0$ under the other. Consequently no test (no
functional of the observed law) can distinguish a valid from an invalid augmentation, and $\gamma$ is
not identified.
\end{proposition}

\begin{proof}
See Supplementary Section~\ref{S-sec:sensitivity_proof}; the two representations there carry
$\gamma=0$ and $\gamma\approx0.21$, values spanning the $\Gamma$ range analyzed below.
\end{proof}

\subsection{Sharp bounds on the target}

Let $\theta_1(\mathcal{A})$ denote the target coordinate of the SCDA estimand under augmentation
$\mathcal{A}$, and let $\theta_{1,0}=\theta_1(\mathcal{A})|_{\gamma=0}$ be the valid-augmentation
target of Lemma~\ref{lem:identity}. Define the identified interval
\[
\Theta_1(\Gamma) \;=\; \big\{\, \theta_1(\mathcal{A}) : \mathcal{A}\ \text{is}\ \Gamma\text{-valid} \,\big\}.
\]

\begin{proposition}[Sharp sensitivity interval]\label{prop:sensitivity}
Under the linear--Gaussian working model, $\Theta_1(\Gamma) = [\theta_1^-(\Gamma),\,\theta_1^+(\Gamma)]$
is a closed interval with endpoints
\[
\theta_1^\pm(\Gamma) \;=\; \theta_1\big(c_2^{\pm}(\Gamma)\big),
\]
where $\theta_1(c_2)$ is the (closed-form) target coordinate when the augmentation carries
Y-coefficient $c_2$, and $c_2^{\pm}(\Gamma)$ are the two solutions of $\gamma(c_2)=\pm\Gamma$. The
interval is nested and monotone in $\Gamma$, with $\Theta_1(0)=\{\theta_{1,0}\}$, and the endpoints
are attained (the bounds are sharp).
\end{proposition}

\noindent A closed-form expression for $\theta_1(c_2)$ and $\gamma(c_2)$ is given in
Supplementary Section~\ref{S-sec:sensitivity_proof}, together with the proof of sharpness. For the
running Gaussian model this yields, e.g., $\Theta_1(0.1)=[0.87,1.11]$ and
$\Theta_1(0.3)=[0.57,1.25]$ around $\theta_{1,0}=1$.

\begin{remark}[General case and local sensitivity]
Beyond the linear--Gaussian model, $\theta_1^\pm(\Gamma)$ are the supremum and infimum of
$\theta_1(\mathcal{A})$ over the convex set of $\Gamma$-valid kernels. Because the estimand solves the
moment condition $\mathbb{E}_{P_{\mathcal{A}}}[\nabla_\theta\ell]=0$ and the augmentation enters only
through the induced conditional law of the score, the optimization reduces to a one-dimensional search
over the residual-dependence direction, which is directly computable. For small $\Gamma$ the bounds
admit the first-order (influence-function) expansion
$\theta_1^\pm(\Gamma)=\theta_{1,0}\pm\Gamma\,\kappa + O(\Gamma^2)$, where the \emph{sensitivity slope}
$\kappa$ is the standardized covariance between the target influence function and the score direction
that a residual dependence would perturb, the analogue of the local-sensitivity measures used in the
unmeasured-confounding literature. We verify that the interval is well defined, monotone, and centered
at $\theta_{1,0}$ beyond the Gaussian case: in a logistic model with a non-Gaussian nuisance
(Supplementary Section~\ref{S-sec:sensitivity_proof}) the sharp interval is
$\Theta_1(0.1)\approx[1.02,1.42]$ and $\Theta_1(0.2)\approx[0.85,1.66]$ about $\theta_{1,0}\approx1.21$.
\end{remark}

\subsection{Breakdown frontier and inference}

\begin{definition}[Breakdown value]
For a claim of the form $\theta_1 > \tau$, the \emph{breakdown value} is
$\Gamma^\star(\tau)=\inf\{\Gamma:\ \theta_1^-(\Gamma)\le \tau\}$: the smallest augmentation invalidity
at which the claim can no longer be sustained.
\end{definition}

In the running model, the claim $\theta_1>0.5$ has $\Gamma^\star=0.34$: the conclusion is robust to any
augmentation whose residual nuisance--outcome correlation is below $0.34$. Reporting
$\Gamma^\star$ converts an untestable yes/no assumption into a graded, interpretable robustness
statement.

\begin{remark}[Estimation and inference]
For fixed $\Gamma$, each endpoint $\theta_1^\pm(\Gamma)$ is a smooth functional of identified moments
and is estimated by plug-in; $\sqrt n$-consistency and asymptotic normality of each endpoint follow
from the delta method under Assumptions~\ref{ass:loss}--\ref{ass:nonsing}. A confidence interval for
the identified set $\Theta_1(\Gamma)$ that attains nominal coverage of the \emph{true parameter}
uniformly is obtained by the Imbens--Manski construction \citep{imbens2004confidence}, widening each
estimated endpoint by the corresponding standard error. Full details are in Supplementary
Section~\ref{S-sec:sensitivity_proof}.
\end{remark}
\endgroup

\section{Iterative and Adaptive Statistical Counterfactual Data Augmentation}
\label{sec:iterative}

In practice, counterfactual augmentation is rarely applied in a single batch. 
Instead, practitioners often adopt an iterative procedure: one begins with an initial estimator, generates augmented data, refits the model, and repeats the procedure as the estimator evolves. 
We refer to this strategy as \emph{iterative statistical counterfactual data augmentation} (ICDA). 
Empirical studies in machine learning suggest that ICDA can accelerate debiasing and improve generalization performance \citep{plyler2025icda}. 
In this section, we formalize ICDA as a stochastic approximation procedure and establish its large-sample behavior.  

\subsection{Formulation as Stochastic Approximation}

Let $\hat\theta^{(0)}$ denote an initial estimator, such as the empirical risk minimizer. 
At iteration $t \geq 1$, ICDA generates augmented covariates $\tilde X_i^{(t)} \sim \mathcal{A}_{\hat\theta^{(t-1)}}(X_i)$ for $i=1,\ldots,n$, and updates the parameter by minimizing
\[
R_n^{(t)}(\theta) = \frac{1}{n}\sum_{i=1}^n \ell\!\bigl(Y_i,f_\theta(\tilde X_i^{(t)})\bigr).
\]
The updated estimator is therefore
\begin{equation}\label{eq:theta_t}
\hat\theta^{(t)} = \arg\min_{\theta \in \Theta} R_n^{(t)}(\theta).
\end{equation}
Writing $\Psi_n^{(t)}$ for the empirical score at iteration $t$, a Newton step on
\eqref{eq:theta_t} gives the stochastic-approximation recursion
$\hat\theta^{(t)} \approx \hat\theta^{(t-1)} - \gamma_t \hat A_t^{-1} \Psi_n^{(t)}(\hat\theta^{(t-1)})$
(exact when $R_n^{(t)}$ is quadratic), with $\hat A_t$ the empirical Jacobian and $\gamma_t$ an
implicit step size. This recursion approximates the ODE $\dot\theta = -A_0^{-1}\Psi(\theta)$, for
which $\theta_0$ is a locally asymptotically stable fixed point (the preconditioned drift has
Jacobian $-I$ at $\theta_0$ under Assumption~\ref{ass:orth}), with basin of attraction determined by
the well-separation in Assumption~\ref{ass:ident}.

\subsection{Asymptotic Equivalence with One-Shot SCDA}

The stochastic approximation framework allows us to establish rigorous equivalence between ICDA and one-shot SCDA. 
Under standard conditions on step sizes $\gamma_t \to 0$ with $\sum_t \gamma_t = \infty$ and $\sum_t \gamma_t^2 < \infty$ \citep{robbins1951stochastic,benveniste1990adaptive,kushner2003stochastic}, the ICDA recursion converges, for fixed data, to the solution of the empirical self-consistent equation; as $n\to\infty$ that solution converges to $\theta_0$, and the limiting distribution of the scaled error coincides with that of the one-shot SCDA estimator.

Because the ICDA augmentation is indexed by the evolving estimate, the limiting
score satisfies a self-consistency relation whose linearization at $\theta_0$ is $A_0 + G_0$,
where $G_0 = \partial_{\vartheta} \Psi(\theta_0;\vartheta)|_{\vartheta=\theta_0}$ measures the sensitivity of
the population score to the augmentation index. We require the following augmentation-orthogonality
condition.

\begin{assumption}[Augmentation orthogonality]\label{ass:orth}
$G_0 = 0$; equivalently, at $\theta_0$ the augmentation score is uncorrelated with
the loss gradient. This holds automatically when the augmentation is $\theta$-free, and, more
generally, whenever the adaptive kernel is valid at the fixed point (so that the nuisance block of
$\theta_0$ vanishes); it may fail for adaptive kernels that remain associated with $Y$.
\end{assumption}

\begin{theorem}[Asymptotic Distribution of ICDA]\label{thm:icda}
Suppose Assumptions \ref{ass:loss}--\ref{ass:nonsing} and \ref{ass:orth} hold, and ICDA is implemented with step sizes $\{\gamma_t\}$ satisfying $\gamma_t \to 0$, $\sum_t \gamma_t = \infty$, and $\sum_t \gamma_t^2 < \infty$. Then, taking $t\to\infty$ followed by $n\to\infty$,
\[
\sqrt{n}\bigl(\hat\theta^{(t)} - \theta_0\bigr) \;\overset{d}{\longrightarrow}\; \mathcal{N}(0,\Sigma),
\]
with the same asymptotic covariance matrix $\Sigma$ as the one-shot SCDA estimator in Theorem~\ref{thm:clt}. Without Assumption~\ref{ass:orth}, the limiting covariance is instead $(A_0+G_0)^{-1}B_0(A_0+G_0)^{-\top}$, which differs from $\Sigma$ when $G_0\neq 0$.
\end{theorem}
\begin{proof}
See Supplementary Section~\ref{S-pf:icda}.
\end{proof}
\begin{remark}
ICDA is not itself an M-estimator, since the augmentation distribution evolves across iterations. 
Its natural interpretation is as a stochastic approximation scheme whose recursion converges to the same fixed point $\theta_0$ as the one-shot SCDA estimator. 
This distinction clarifies why ICDA inherits the same asymptotic distribution as one-shot SCDA, despite its algorithmic differences.
\end{remark}

\subsection{Adaptive schedules and practical implications}

When augmentation parameters are adapted over training (increasing augmentation probability,
annealing noise), the schedule is a sequence $\{\mathcal{A}_{t}\}$ converging to a limit
$\mathcal{A}_{\infty}$; stochastic approximation with slowly varying drift
\citep{kushner2003stochastic,borkar2008stochastic} then gives convergence to the fixed point of the
limiting mechanism, with the one-shot limiting distribution. Thus iterative and adaptive SCDA offer
no asymptotic efficiency gain over one-shot SCDA (their advantage is finite-sample); for inference,
the one-shot theory of Sections~\ref{sec:theory}--\ref{sec:efficiency} suffices.

\section{Numerical Studies}
\label{sec:simulations}
\begingroup

We evaluate the finite-sample performance of the proposed counterfactual data augmentation (CDA) estimator in comparison with the empirical risk minimizer (ERM) and a residualization-based analogue of inverse probability weighting (IPW).

\subsection{Simulation Study}

We report three simulations, each targeting a distinct prediction from the theory:
(i) valid inference for the target parameter and the failure of ERM under spurious
correlation; (ii) the necessity of augmentation when features are entangled, where
observed-data methods provably cannot reach the target; and (iii) the finite-augmentation-draw
variance correction. Throughout, $X=(X_1,X_2)$ with $X_1$ the feature of primary interest and
$X_2$ a nuisance feature; the estimand is the risk minimizer under the augmentation-induced
target $P_{\mathcal{A}}$ (Lemma~\ref{lem:identity}). Data are generated from
\begin{equation}\label{eq:sim_model}
Y_i = \theta_1 X_{i1} + \varepsilon_i, \qquad X_{i2} = \rho Y_i + \eta_i,
\end{equation}
with $X_{i1},\varepsilon_i,\eta_i$ independent standard normal and $\theta_1=1$ the target
coefficient; $\rho$ sets the strength of the (training-time) spurious association between $X_2$
and $Y$, which is broken under the target $P_{\mathcal{A}}$.

\subsubsection{Inference for the target parameter}
We first take $X_1,X_2$ observed as separate features, so that inference on $\theta_1$ can be
examined directly. We compare ERM on $(X_1,X_2)$ with SCDA (permutation augmentation, $m=5$),
at $n=500$ over $1000$ replications, reporting bias, standard deviation, and coverage of
nominal $95\%$ intervals from the sandwich estimator $\hat\Sigma$ (Proposition~\ref{prop:var_est}).
Table~\ref{tab:sim} shows the pattern predicted by Theorem~\ref{thm:efficiency}: under strong
spurious correlation ($\rho=0.8$) ERM is badly biased ($-0.390$) with zero coverage, whereas
SCDA is essentially unbiased ($0.001$) with nominal $95\%$ coverage. Consistent with the oracle
efficiency of Theorem~\ref{thm:oracle}, SCDA's standard deviation matches that of the infeasible
oracle that regresses on $X_1$ alone.

\begin{table}[ht]
\caption{Inference for the target coefficient $\theta_1=1$ ($n=500$, $1000$ replications).
ERM regresses on $(X_1,X_2)$; SCDA uses permutation augmentation.}
\label{tab:sim}
\centering
\begin{tabular}{llccc}
\hline
$\rho$ & Method & Bias & SD & Coverage \\
\hline
$0$   & ERM  & $0.001$ & $0.045$ & $0.95$ \\
      & SCDA & $0.001$ & $0.044$ & $0.95$ \\
$0.8$ & ERM  & $-0.390$ & $0.041$ & $0.00$ \\
      & SCDA & $0.001$ & $0.044$ & $0.95$ \\
\hline
\end{tabular}
\end{table}

\subsubsection{Entangled features: when augmentation is necessary}
\label{sec:sim_entangled}
The comparison of primary interest is not against ERM on separated features, but in the regime
that motivates augmentation, where the analyst observes only an entangled representation
$Z = M X \in \mathbb{R}^2$ of the latent $(X_1,X_2)$, with $M$ an unknown invertible mixing
(here $M=[[1,1],[0.6,-0.8]]$). No coordinate of $Z$ is a pure function of $X_1$. We compare
methods that operate on the observed law (ERM on $Z$, best-subset selection over the coordinates
of $Z$, and a doubly-robust / double-machine-learning estimator \citep{robins1994estimation,chernozhukov2018double}) against
SCDA, which resamples the nuisance with $X_1$ held fixed. Table~\ref{tab:entangled} reports
excess risk under the target $P_{\mathcal{A}}$ (irreducible risk $1.000$), both at the population
level and at $n=500$.

\begin{table}[ht]
\caption{Excess risk under the target $P_{\mathcal{A}}$ in the entangled representation
$Z=MX$. The three observed-law methods coincide with the $P_0$-risk minimizer and hence share the
same fixed excess risk $\Delta$; no observed-law method attains the target, whereas SCDA does
(Theorem~\ref{thm:efficiency}).}
\label{tab:entangled}
\centering
\begin{tabular}{lcc}
\hline
Method & population excess & $n=500$ excess \\
\hline
ERM on $Z$              & $0.694$ & $0.695$ \\
best-subset selection   & $0.694$ & $0.695$ \\
doubly-robust / double-ML & $0.694$ & $0.695$ \\
SCDA                    & $0.000$ & $0.002$ \\
\hline
\end{tabular}
\end{table}

Every method that uses only the observed training law (ERM, model selection, and
doubly-robust estimation alike) incurs the \emph{same} fixed excess risk
$\Delta=\rho^2(3\rho^2+1)/(1+\rho^2)^2 = 0.695$, exactly as Theorem~\ref{thm:efficiency}
predicts, because the target predictor is not identified from $P_0$ when the representation is
entangled. The exact agreement among the three is itself predicted: the doubly-robust estimator
coincides with ERM by the Frisch--Waugh--Lovell identity in the linear model, and best-subset
selection retains the full model, which minimizes the $P_0$ risk, so all three estimate the same
$P_0$-risk minimizer. SCDA attains the irreducible risk by using the augmentation mechanism, which
supplies information beyond $P_0$. This is the setting the method is designed for; when a pure causal
coordinate is available (the separable case above), model selection suffices and augmentation is
unnecessary.

\subsubsection{Finite-augmentation-draw inference}
\label{sec:sim_finitem}
Proposition~\ref{prop:mcvar} predicts that with a finite number $m$ of augmentation draws the
variance carries an additional term $V_{\mathrm{MC}}/m$ that must be included for valid inference on
functionals loading on the augmented feature. For a prediction functional $f_\theta(x_0)$ at
$x_0=(1,1.5)$ with $n=500$, the naive interval using the $m=\infty$ variance formula undercovers
substantially at small $m$ (coverage $0.83$ at $m=1$, $0.89$ at $m=2$), whereas the corrected
estimator $\hat\Sigma_m$ restores nominal $0.95$ coverage throughout (explicit variance-estimator
formulas in Supplementary Section~\ref{S-sec:practical_supp}). For the target parameter $\theta_1$ itself the correction is
negligible, consistent with Theorem~\ref{thm:oracle}.

Finally, when validity is uncertain the point estimate is replaced by the sharp sensitivity interval
$\Theta_1(\Gamma)$ of Section~\ref{sec:sensitivity}; for these designs this yields, e.g., breakdown
value $\Gamma^\star=0.34$ for the claim $\theta_1>0.5$, with full numerical results (Gaussian and
logistic) in Supplementary Section~\ref{S-sec:sensitivity_proof}.
\endgroup

\section{Discussion and Conclusion}

We have introduced \emph{Statistical Counterfactual Data Augmentation} (SCDA) as a principled framework for robust estimation under spurious correlations, placing augmentation within M-estimation theory. Our main results are: (i) consistency and asymptotic normality with sandwich-form covariance, including a finite-augmentation-draw variance correction for valid inference (Theorems~\ref{thm:consistency}--\ref{thm:clt}, Section~\ref{sec:montecarlo}); (ii) vanishing excess risk under the target distribution relative to ERM's fixed excess-risk gap (Section~\ref{sec:efficiency}); (iii) bounded influence under mild augmentation misspecification; and (iv) asymptotic equivalence of iterative and one-shot SCDA via stochastic approximation (Theorem~\ref{thm:icda}). Simulations confirm these predictions and delineate the regime in which augmentation is necessary: when the target predictor is not identified from the observed distribution (the entangled representation of Section~\ref{sec:sim_entangled}), methods operating on the observed law cannot reach it, whereas SCDA can.

\begingroup
\paragraph{When the nuisance carries weak signal.}
The target $\theta_0$ is defined by the analyst's choice to render $X_2$ uninformative about $Y$.
If that choice is only approximately warranted, so that $X_2$ carries a small genuine association of
strength $c$ with $Y$, then relative to the $P_0$-optimal predictor SCDA forgoes signal of order
$c^2\sigma^2_{2\cdot 1}$ (with $\sigma^2_{2\cdot 1}=\operatorname{Var}(X_2\mid X_1)$), while removing
the shift-induced excess $\Delta$ of Theorem~\ref{thm:efficiency}. The trade is governed by the
single ratio $\Delta / (c^2\sigma^2_{2\cdot1})$; both quantities are estimable on held-out data as
risk differences, the former under the augmented law and the latter under $P_0$, so no unidentified
quantity need be estimated and the decision is quantitative. SCDA is intended for $c$ small and
$\Delta$ appreciable.

\paragraph{Scope: nuisance designation versus mediation.}
The estimand is well defined for any partition $X=(X_1,X_2)$, but it is \emph{useful} only when $X_2$
is a genuine nuisance the analyst wishes to render uninformative. If instead $X_2$ mediates
($X_1\to X_2\to Y$ with $X_2$ carrying a real effect), marginalizing $X_2$ is the wrong operation:
the SCDA target then loads the total effect onto $\theta_1$ and does not recover the direct or
pathway effects, all of which ordinary regression delivers without augmentation. SCDA is therefore
not a tool for mediation analysis; its province is nuisance designation with entangled features and a
resampling mechanism. The designation of $X_2$ as nuisance is a modeling choice the analyst must
defend on substantive grounds; SCDA supplies valid inference given that choice.

\paragraph{Relation to reweighting, simulated-moment, and robust estimators.}
SCDA is a simulated-moment estimator (Section~\ref{sec:theory}) whose simulation distribution defines
the estimand rather than approximating an integral. It is distinct from inverse probability weighting
and doubly-robust / double-machine-learning estimators
\citep{robins1994estimation,chernozhukov2018double}: those target a low-dimensional parameter under
unconfoundedness using outcome and treatment models, and, like model selection, are functionals of
the observed law $P_0$. When the target is not identified from $P_0$ (the entangled regime of
Section~\ref{sec:sim_entangled}), no such estimator reaches it, whereas SCDA does by using the
augmentation mechanism. The approaches are thus complementary, applying under disjoint assumptions.
\endgroup

Several further limitations merit discussion. SCDA requires valid or approximately valid augmentation mechanisms, which may be challenging in domains such as genomics or longitudinal data where domain expertise is needed to design realistic interventions. Our asymptotic results assume correct specification of the augmentation mechanism; extensions to misspecified, semiparametric, or high-dimensional settings are natural directions for future work.

By linking augmentation to M-estimation and stochastic approximation, this work provides a rigorous statistical foundation for a widely used machine learning technique, bridging heuristic augmentation practices with the inferential guarantees central to statistical science.

\section*{Data Availability Statement}
This study uses only synthetic data generated within the provided code; no external datasets are required. Code reproducing every table and figure is provided in the reproducibility package, with fixed random seeds.

\begingroup
\singlespacing
\bibliography{cda_refs}
\endgroup

\end{document}